\documentclass[twocolumn,showpacs,preprintnumbers,superscriptaddress,amsmath,amssymb,nofootinbib]{revtex4}
\input psfig.sty
\usepackage{graphicx}
\usepackage[mathscr]{eucal}
\usepackage{amssymb}
\usepackage{setspace}
\usepackage{color}

\begin{document}

\title{The unphysical character of dark energy fluids}

\author{Ronaldo C. Duarte} \email{ronaldocesar@uern.br}
\affiliation{Departamento de Matem\'atica e Estat\'{i}stica, Universidade do Estado do Rio Grande do Norte, 59610-210, Mossor\'o, RN, Brazil}
\author{Ed\'esio M. Barboza Jr.} \email{edesiobarboza@uern.br}
\affiliation{Departamento de F\'isica, Universidade do Estado do Rio Grande do Norte, 59610-210, Mossor\'o, RN, Brazil}
\author{Everton M. C. Abreu}\email{evertonabreu@ufrrj.br}
\affiliation{Grupo de F\' isica Te\'orica e F\' isica Matem\'atica, Departamento de F\'{i}sica, Universidade Federal Rural do Rio de Janeiro, 23890-971, Serop\'edica, RJ, Brazil}
\affiliation{Departamento de F\'{i}sica, Universidade Federal de Juiz de Fora, 36036-330, Juiz de Fora, MG, Brazil}
\affiliation{Programa de P\'os-Gradua\c{c}\~ao Interdisciplinar em F\'isica Aplicada, Instituto de F\'{i}sica, Universidade Federal do Rio de Janeiro, 21941-972, Rio de Janeiro, RJ, Brazil}
\author{Jorge Ananias Neto}\email{jorge@fisica.ufjf.br}
\affiliation{Departamento de F\'{i}sica, Universidade Federal de Juiz de Fora, 36036-330, Juiz de Fora, MG, Brazil}

\newtheorem{definition}{\bf Definition}
\newtheorem{affirmation}{\bf Affirmation}
\newtheorem{theorem}{Theorem}
\newtheorem{remark}{Remark}
\newtheorem{lemma}{Lemma}
\newtheorem{proposition}{Proposition}
\newtheorem{corollary}{Corollary}
\newtheorem{proof}{Proof}
\newtheorem{proofa}{Proff of Affirmation}
\def\ni{\noindent}
\def\no{\nonumber \\}

\date{\today}

\begin{abstract}

\ni {\color{black} It is well known that, in the context of general relativity, an unknown kind of matter that must violate the strong energy condition is required to explain the current accelerated phase of expansion of the Universe. This unknown component is called dark energy and is characterized by an equation of state parameter $w=p/\rho<-1/3$. Thermodynamic stability requires that $3w-d\ln |w|/d\ln a\ge0$ and positiveness of entropy that $w\ge-1$. In this paper we proof that we cannot obtain a differentiable function $w(a)$ to represent the dark energy that satisfies these conditions trough the entire history of the Universe. 

}

\end{abstract}


\maketitle

\section{Introduction}

Between the end of the twenty century and the beginning of the twenty one's, a large amount of observational data revealed that the Universe is currently expanding at an accelerated rate \cite{acc_exp1,acc_exp2,acc_exp3,acc_exp4}. This accelerated expansion can be easily explained if a cosmological constant is added to the Einstein field equations. It can be shown that the zero point energy of all fields filling the Universe acts into the Einstein equations as a cosmological constant \cite{zeldovich}. However, the observational constraints on the cosmological constant differs from the theoretical value provided by quantum field theory by at least $60$ orders of magnitude \cite{weinberg, padmanabhan, bousso}. Although many proposals to solve this problem appeared in the literature \cite{CCP1,CCP2,CCP3,CCP4,CCP5,CCP6,CCP7,CCP8} none of them provide a real conclusive solution to the problem.  The lack of a convincing explanation for the cosmological constant problem has led physicists to adopt a pragmatic approach which assume that the cosmological constant is canceled out by some unknown symmetries in nature and that the accelerated expansion is due to some unknown kind of matter. In order to obtain an accelerated expansion of the Universe at present time, the matter content of the Universe must violates the strong energy condition, i. e.,

\begin{equation}
\label{SEC_violation}
\sum_{i}\Big(\rho_{i}+3\frac{p_i}{c^2}\Big)<0\,\,,
\end{equation}

\noindent where $\rho_i$ and $p_i$ are, respectively, the energy density and the pressure of the $i-$th component of the matter content of the Universe and $c$ the speed of the light. The violation of the strong energy condition implies that the total pressure must be negative. Since baryonic and cold dark matter are pressureless ($p_m=0$), and the pressure of relativistic matter is $\rho_{\gamma}c^2/3$, the Universe must contain an additional source term with a pressure sufficiently negative to ensure the validity of (\ref{SEC_violation}). This unknown source was dubbed dark energy (DE) (see Ref. \cite{copeland} for a review). DE is frequently characterized by the equation of state (EoS) parameter $w=p/\rho c^2$ which mimics a cosmological constant if $w=-1$, a quintessence scalar field if $-1\le w\le 1$ \cite{quintessence}, a phantom field  if $w<-1$ \cite{phantom} and many other forms of exotic matter.

Since most DE models are able to adjust the data seamlessly, it is extremely difficult to decide which of these models is correct, if there is one. The data shows only that the Universe is expanding at an accelerated rate but does not reveal which object causes this acceleration, i.e., if it is DE or something else. However, although unknown, DE is not immune to the laws of physics. Such an exotic fluid must satisfy the bounds imposed by the laws of thermodynamics which have a strong experimental basis. Here we investigate the limits imposed by thermodynamics to a DE fluid. We proof that thermodynamics rule out DE fluids.

This paper is organized as follows: Section II summarizes the main results of the thermodynamical of cosmic fluids contained in Ref. \cite{Barboza_thermo}; Section III contains the proof that there is no EoS parameter $w(a)$ of DE satisfying the thermodynamical through the entire history of the Universe: there is at least one point of discontinuity at $0<a(t)<1$ where $w(a)$ is not a differentiable function and where the stability condition fails. Section IV contains our conclusion.

\section{Thermodynamics of the cosmic fluids}

{\color{black} In this section we will analyze the thermodynamical properties of the function that describes DE cosmic fluids.   The objective is to construct, based on these ``heat" properties, a general function to analyze afterwards, its mathematical viability.}

\subsection{Internal energy and entropy}

Let us consider an expanding, homogeneous and isotropic Universe filled by $n$ no interacting perfect fluids. Since all physical distances scale with the same factor $a(t)$, the physical volume of the Universe at a given time is $V=a^3(t)V_0$\footnote{Here the index $0$ will denote the present time value of an observable and we will adopt the convention $a_0=1$}. In such a model the internal energy of the $i$th fluid  component can be written as
\begin{equation}
\label{FRW_energy}
U_i=\rho_i c^2 V\,\,.
\end{equation}

\noindent Assuming a reversible adiabatic expansion, 
the first law of thermodynamics
\begin{eqnarray}
\label{Gibbs_law}
T_idS_i&=&dU_i+p_idV,
\end{eqnarray}

\noindent leads to so-called fluid equation,
\begin{equation}
\label{fluid_eq}
d\ln\rho_i+3(1+w_i)d\ln a=0\,\,,
\end{equation}

\noindent which expresses the energy-momentum conservation. Assuming that the density is a function of the temperature and volume, i. e., $\rho_i=\rho_i(T_i, V)$, the fact that $dS_i$ is an exact differential implies that \cite{weinberg1971}
\begin{equation}
\label{temperature_law_w}
d\ln T_i=-3w_id\ln a\,\,,
\end{equation}

\noindent or, using Eq. (\ref{fluid_eq}) to eliminate $w_i$,
\begin{equation}
\label{temperature_law_rho}
d\ln T_i=d\ln \rho_i+ 3d\ln a\,\,.
\end{equation}

\noindent Integration of Eq. (\ref{temperature_law_rho}) provides
\begin{equation}
\label{T_rho_rel}
\frac{T_i}{T_{i,0}}=\frac{\rho_i}{\rho_{i,0}}a^3\,\,.
\end{equation}



\noindent Thus, the internal energy of the $i$th fluid component (\ref{FRW_energy}) can be written as
\begin{equation}
\label{internal_energy}
U_i=U_{i,0}\frac{T_i}{T_{i,0}}\,\,.
\end{equation}

The entropy of the $i$th fluid component is obtained from the Euler relation \cite{entropy_positiveness}:
\begin{equation}
\label{Euler}
U_i=T_iS_i-p_iV+\mu_iN_i\,\,,
\end{equation}

\noindent where $\mu_i$ and $N_i$ are, respectively, the chemical potential and the number of particles of the $i$th component. By assuming that the chemical potential is zero we obtain, combining Eqs. (\ref{FRW_energy}), (\ref{T_rho_rel}) and (\ref{Euler}), that
\begin{equation}
\label{entropy}
S_i=(1+w_i)\frac{\rho_{i,0}c^2V_0}{T_{i,0}} \,\,,
\end{equation}

\noindent {\color{black} which shows the direct relation between the entropy and the EoS parameter, $w_i$, for a given component of the cosmic fluid.}





%

\subsection{Heat capacity}

The classical thermodynamical definition of a fluid's heat capacity $C_i$ is \cite{callen},
\begin{equation}
\label{SH_definition}
dQ_i=C_idT_i\,\,,
\end{equation} 

\noindent where $dT_i$ is the fluid temperature increase due to an absorbed heat $dQ_i=T_idS_i$. The heat capacity of a fluid will differ depending on whether the fluid is heated at constant volume or at constant pressure. From the first law of thermodynamics, Eq. (\ref{Gibbs_law}), at constant volume, Eq. (\ref{SH_definition}) becomes
\begin{equation}
\label{CV0}
dU_i=C_{iV}dT_i\,\,,
\end{equation}

\noindent where
\begin{equation}
\label{CV}
C_{iV}=\Big(\frac{\partial U_i}{\partial T_i}\Big)_V\,\,,
\end{equation}

\noindent is the fluid's heat capacity at constant volume. From Eq. (\ref{internal_energy}), it is easy to show that
\begin{equation}
\label{CV_cosmic}
C_{iV}=\frac{U_{i,0}}{T_{i,0}}={\rm constant}\,\,,
\end{equation}

\noindent for any component of the Universe.

Now, from the enthalpy definition,
\begin{equation}
\label{enthalpy}
h_i=U_i+p_iV\,\,,
\end{equation}

\noindent the first law of thermodynamics can be written as
\begin{equation}
\label{1st_law_enthalpy}
dQ_i=dh_i-Vdp_i\,\,.
\end{equation}

\noindent Therefore, at constant pressure, Eq. (\ref{SH_definition}) becomes
\begin{equation}
\label{CP0}
dh_i=C_{p_i}dT_i\,\,,
\end{equation}

\noindent where
\begin{equation}
\label{CP}
C_{p_i}=\Big(\frac{\partial h_i}{\partial T_i}\Big)_{p_i}\,\,,
\end{equation}

\noindent is the fluid's heat capacity at constant pressure. Since $p_iV=w_iU_i$, the enthalpy Eq. (\ref{enthalpy}) becomes 
\begin{equation}
h_i=(1+w_i)U_i
\end{equation}

\noindent and, from Eqs. (\ref{internal_energy}) and (\ref{temperature_law_w}), we have that
\begin{equation}
\label{Cp_time_dependent}
C_{p_i}=\Big(1+w_i-\frac{1}{3}\frac{d\ln \vert w_i\vert}{d\ln a}\Big)C_{iV}\,\,.
\end{equation}

\noindent {\color{black} which shows that we can write a compact relation between the heat capacities at constant pressure and constant volume, something like $C_{p_i}=\Omega(w_i) C_{iV}$, where $\Omega(w_i)$ is a function of the EoS parameter, defined in the last equation.}

%

\subsection{Compressibility and expansibility}

\noindent By considering the volume as function of temperature and pressure, we have that\footnote{Remember that we are assuming that the fluids evolve separately, that is, they do not exchange heat, as shown in Eq. (\ref{fluid_eq}).}
\begin{eqnarray}
\label{vol_diff}
dV&=&\sum_i\Big[\Big(\frac{\partial V}{\partial T_i}\Big)_{p_i}dT_i+\Big(\frac{\partial V}{\partial p_i}\Big)_{T_i}dp_i\Big]\nonumber\\
    &=&V\sum_i(\alpha_idT_i-\kappa_{T_i}dp_i) \,\,,
\end{eqnarray}

\noindent where 
\begin{equation}
\label{alpha}
\alpha_i\equiv\frac{1}{V}\Big(\frac{\partial V}{\partial T_i}\Big)_{p_i}
\end{equation}

\noindent is the thermal expansibility and  
\begin{equation}
\label{kappa_t}
\kappa_{T_i}\equiv-\frac{1}{V}\Big(\frac{\partial V}{\partial p_i}\Big)_{T_i}
\end{equation}

\noindent is the isothermal compressibility. The thermal expansibility measures the thermal volume expansion at constant pressure and isothermal compressibility measures the relative modification of the volume together with the increasing pressure at fixed temperature. 

Analogously to the isothermal compressibility, we can define the adiabatic compressibility $\kappa_{S_i}$ if, instead of temperature, the entropy is kept fixed. 
%
%
\noindent It can be shown that the isothermal compressibility and the isothermal expansibility are related by
\begin{equation}
\label{ak_ratio}
\frac{\alpha_i}{\kappa_{T_i}}=\Big(\frac{\partial p_i}{\partial T_i}\Big)_V\,\,,
\end{equation}

\noindent and that the ratio between the adiabatic and the isothermal compressibilities are equal to the ratio between the  heat capacities at constant volume and at constant pressure, i.e.,
\begin{equation}
\label{kk_ratio}
\frac{\kappa_{S_i}}{\kappa_{T_i}}=\frac{C_{iV}}{C_{p_i}}\,\,.
\end{equation}

\noindent Notice that $p_iV=w_iC_{iV}T_i$ and using Eq. (\ref{temperature_law_w}) we obtain
\begin{equation}
\label{a_time_dependent}
\alpha_i=\frac{C_{iV}}{p_iV}\Big(w_i-\frac{1}{3}\frac{d\ln \vert w_i\vert}{d\ln a}\Big)\,\,.
\end{equation}

\noindent From (\ref{ak_ratio}) is easy to show that
\begin{equation}
\label{27}
\kappa_{T_i}=\frac{\alpha_i V}{w_iC_{iV}}\,\,,
\end{equation}

\noindent and from the above equation and (\ref{kk_ratio}) we have 
\begin{equation}
\label{28}
\kappa_{S_i}=\frac{\alpha_i V}{w_iC_{p_i}}\,\,
\end{equation}

\noindent {\color{black} and, differently from Eq. \eqref{entropy}, that shows a direct relation between the entropy and the EoS parameter $w_i$, Eq. \eqref{28} shows a much more intricate relation between the adiabatic compressibility and $w_i$   Substituting Eq. \eqref{a_time_dependent} into Eqs. \eqref{27} and \eqref{28} we have 
\begin{equation}
\label{29}
\kappa_{T_i} \,=\, \frac{1}{w_i^2\rho_i}\, \Big(w_i\,-\,\frac{1}{3}\frac{d\,ln|w_i|}{d\,ln a} \Big)
\end{equation}

\noindent and

\begin{equation}
\label{30}
\kappa_{S_i} \,=\, \frac{C_{iV}}{w_i^2\rho_i C_{p_i}}\, \Big(w_i\,-\,\frac{1}{3}\frac{d\,ln|w_i|}{d\,ln a} \Big)
\end{equation}

\noindent which can be used to determine the constraints on cosmic fluids EoS parameter $w_i$.   We will define some of these constraints just below.}


\subsection{Constraints on cosmic fluids}

Thermodynamical stability requires that $C_{iV},C_{p_i},\kappa_{S_i},\kappa_{T_i}\ge0$ simultaneously. Conversely, these quantities are all negative simultaneously if the stability cannot be obtained. From Eq. (\ref{CV_cosmic}) it is clear that $C_{iV}\ge0$ so that the cosmic fluids satisfies the stability conditions. Moreover, $C_{iV},C_{p_i},\kappa_{S_i},\,{\rm and}\,\kappa_{T_i}$ are related by
$$
C_{p_i}=C_{iV}+\frac{TV\alpha_i^2}{\kappa_{T_i}}
$$
\noindent and
$$
\kappa_{T_i}=\kappa_{S_i}+\frac{TV\alpha_i^2}{C_{p_i}}
$$
\noindent so that $C_{p_i}\ge C_{iV}$ and $\kappa_{T_i}\ge\kappa_{S_i}$. From Eqs. (\ref{29}) and (\ref{30}) it is easy to see that the above conditions are satisfied only if the fluid EoS parameter obeys the constraint
\begin{equation}
\label{thermo_bound}
w_i-\frac{1}{3}\frac{d\ln \vert w_i\vert}{d\ln a}\ge0.
\end{equation}

\noindent Along with the above constraint, the positiveness of entropy implies that 
\begin{equation}
\label{S>0}
w_i\ge-1 \,\,.
\end{equation}

\noindent It is obvious that if $w_i$ is constant, thermodynamical stability implies that $w_i\ge0$ which rule out all negative pressure fluids with a constant EoS parameter. 

In order to accelerate the Universe at present time is also required by (\ref{SEC_violation}) that $w(a\to1)<-1/3$. 

\section{Thermodynamical inviability of dark energy fluids}

In this section we will investigate the existence of functions $w(a)$ that satisfies the following conditions:

\begin{enumerate}
\item[i)] $\displaystyle{w-\frac{1}{3}\frac{d\ln \vert w\vert}{d\ln a}\ge0}$;
\item[ii)] $w\ge-1$ for all $a>0$;
\item[iii)] $ \lim\limits_{a \rightarrow 1}w(a)=w_0<-\frac{1}{3}$. 
\end{enumerate}

\noindent Conditions i) and ii) are thermodynamical constraints on all cosmic fluids physically acceptable, and condition iii) is required to accelerate the Universe at present time. 
Below we will investigate the existence of differentiable functions, except an a finite set, i. e., 
functions $w$ such that the set $\left\{t; w \mbox{ is not differentiable at t}\right\}$ is finite, whose derivative is continuous for $t$ sufficiently small and that satisfies the condition i), ii) and iii) above. We will show that there is no function $w(a)$ of $C_1$ class satisfying the conditions i), ii) and iii).
In our proof, we will use the follow theorems (see \cite{elonlages} for the proofs):

\begin{theorem}
	Let $f:[a,b] \rightarrow \mathbb{R}$ be continuous. If $f(a)<d<f(b)$ then there is $c \in (a,b)$ such that $f(c)=d$.
\end{theorem}

\begin{theorem}
	Let $f:[a,b] \rightarrow \mathbb{R}$ be continues. If $f$ is differentiable at $(a,b)$, then there is $c \in (a,b)$ such that $$f'(c)=\frac{f(b)-f(a)}{b-a}.$$
\end{theorem}

\begin{theorem}
	Let $f,g:[a,b] \rightarrow \mathbb{R}$ be integrable functions. If $f(x)\leq g(x)$ for all $x \in [a,b]$, then
	$$
	\int_{a}^{b}f(x)dx \leq \int_{a}^{b}g(x)dx
	$$
\end{theorem}

\begin{theorem}
	Let $f,g:X \rightarrow \mathbb{R}$. If $f(x) \leq g(x)$ for all $x \in X$ and $\lim\limits_{x \to a}g(x)=-\infty$ then $\lim\limits_{x \to a}f(x)=-\infty.$
\end{theorem}

%
%
 
 \begin{proposition}
 Let $w:(0, + \infty)\rightarrow \mathbb{R}$ be a function and suppose that there is $\delta>0$ such that $w$ is differentiable and its derivative is continuous at $(0,\delta)$ with
 
\begin{enumerate}
\item[i)]  $\displaystyle{w-\frac{1}{3}\frac{d \ln|w|}{d\ln a} \geq 0},$ where $\vert w\vert$ is differentiable.
\end{enumerate}
If there is $t_{0} \in (0, \delta)$ such that $w(t_{0})<0$ then $\lim\limits_{t \rightarrow 0}w(t)= - \infty.$
\end{proposition}

\begin{proof}
Let $\delta>0 $ with $w$ differentiable and with a continuous derivative at $(0, \delta)$. Suppose that there is $t_{0} \in (0, \delta)$	such that $w(t_{0})<0$. Let $c<0$  such that $w(t_{0})<c<0$. $w$ is continuous at $(0, \delta)$, since it is differentiable in this interval and, therefore, we can suppose that at $t \in (t_{0}-\epsilon, t_{0}+\epsilon)$  we have $w(t)<c$ for $\epsilon$ with $t_{0}+ \epsilon< \delta$. We affirm that $w(t)\leq c$ 
for all $t \in (0, t_{0}+\epsilon)$. In fact, otherwise we would take $\overline{t} = \sup \left\{t \in (0,t_{0}+ \epsilon); w(t)> c\right\}$ and we would have 
\begin{enumerate}
\item $\overline{t}\leq t_{0}- \epsilon$;
\item and, by supremum definition, that at $ t \in(\overline{t},t_{0}+ \epsilon)$, $w(t)\leq c<0$. Therefore 
$$
\frac{dw}{dt}>0
$$
in this interval. But, by the Theorem 2, there would exist $s \in (\overline{t}, t_{0})$ such that
$$
\frac{dw}{dt}(s)= \frac{w(t_{0})-w(\overline{t})}{t_{0}- \overline{t}}<0
$$
and that would be contradictory.
\end{enumerate}
In short we have $w(t)\leq c<0$ for all $t \in (0,t_{0}+\epsilon)$ and $(0, t_{0}+\epsilon) \subset (0, \delta)$, that is, $w$ is differentiable at $(0,t_{0}+\epsilon)$. It follows that
$$
\frac{dw}{dt} \geq \frac{w^{2}(t)}{3t} \geq \frac{c^{2}}{3t}.
$$
From the fundamental theorem of calculus and Theorem 3, for $0<s<t_{0}$
$$
w(t_{0})-w(s) \geq \frac{c^{2}}{3}[\ln(t_{0})- \ln(s)]\,\,.
$$
From this inequality we can conclude that
$$
\lim\limits_{s \rightarrow 0}w(s)= - \infty.
$$
\end{proof}
\begin{corollary}
	Let $w:(0, + \infty)\rightarrow \mathbb{R}$ be a differentiable function and with continuous derivative at $(0,\delta)$ satisfying:
	
	\begin{enumerate}
		\item[i)]   $\displaystyle{w-\frac{1}{3}\frac{d \ln|w|}{d\ln a} \geq 0};$
		\item[ii)] $-1\leq w(a)$ for all $a>0$\,\,.
	\end{enumerate}
Then, $w(t)\geq0$ for all $t \in (0, \delta).$
\end{corollary}

\begin{proof}
If there were $t_{0} \in (0, \delta)$ such that $w(t_{0})<0$, then by the previous proposition we wold have that $\lim\limits_{s \rightarrow 0}w(s)= - \infty$. However from hypothesis 2 above this could not occur. 
\end{proof}


\begin{corollary}
Let $w:(0, + \infty) \rightarrow \mathbb{R}$ be a differentiable function at $(0, \delta)$  and with continuous derivative in this interval for some $\delta<1$. Additionally, suppose that the set $\left\{t \in (0,1); w \mbox{ non differentiable at t}\right\}$ is finite and that
\begin{enumerate}
\item[i)]  $\displaystyle{w-\frac{1}{3}\frac{d \ln|w|}{d\ln a} \geq 0};$ where $|w|$ is differentiable.

\item[ii)] $-1\leq w(a)$ for all $a>0$.

\item[iii)] $ \lim\limits_{a \rightarrow 1}w(a)=w_0$ for some $-1\leq w_0<0$.
\end{enumerate}
Then $w$ is discontinuous at some point $t \in (0, 1)$.
\end{corollary}

\begin{proof}
Suppose that there is a $w$, differentiable at $(0,\delta)$ with continuous derivative in this interval, that it satisfies i), ii) and iii) and that $w$ is continuous at $(0,1)$. From Corollary $1$ there is $t_{0}<1$ such that $w(t_{0})\ge0$. Let us choose $d \in \mathbb{R}$ such that 
\begin{equation}\label{eq0}
0>d>w_0\,\,,
\end{equation}
 where $w_0$ is given by hypothesis iii). From Theorem 1, there is a $t_{1} \in (t_{0},1)$ such that $w(t_{1})=d$. Let us define
$$
A= \left\{t \in (0,1); w(t)\geq d\right\}\,\,.
$$
The set $A$ is not empty, given that $t_{1} \in A$. Let 
$$
\overline{t}= \sup A
$$ 
($\sup A$ is the supremum of the set $A$).
Using the continuity condition and iii) there is $\epsilon>0$ such that for all $t \in (1-\epsilon,1)$, $w(t)<d$. This imply that $\overline{t}\leq 1-\epsilon<1$. For all $t \in (\overline{t},1)$ we have $w(t)<d<0$, in particular, by i)
\begin{equation}\label{eq1}
\frac{dw}{dt}\geq \frac{w^{2}}{3a}>0
\end{equation}
Let $t_{1}, t_{2},..., t_{n}$ be the points of the interval $(\overline{t},1)$ such that $w$ is not differentiable. Let us suppose, without lost of generality, that $t_{0}=\overline{t}< t_{1}< t_{2}<,...,< t_{n}<1=t_{n+1}$. Note that, by using Eq. (\ref{eq1}), the function $w$ is an increasing one at $(t_{i},t_{i+1})$ for all $i=0,1,2,...,n+1$. Therefore, for all $s,t \in (t_{i},t_{i+1})$
 we have $w(s)<w(t)$ if $s<t$. From the continuity of $w$ at $t_{i}$ we have
 $$
 w(t_{i})= \lim\limits_{s \rightarrow t_{i}^{+}}w(s)\leq w(t)\,\,,
 $$

\noindent and from the continuity of $w$ at $t_{i+1}$ we have
\begin{equation}\label{eq2}
w(t_{i})\leq \lim\limits_{t \rightarrow t_{1}}w(t)=w(t_{i+1})\,\,,
\end{equation}
for all $i=0,1,2,...n$.
By the inequalities in Eq. (\ref{eq2}), the supremum definition and the continuity of $w$, we have that
$$
d \leq w(\overline{t}) \leq w(t_{1})\leq w(t_{2}) \leq... \leq w(t_{n-1})\leq w(t_{n})\leq w_0\,\,.
$$
This inequality is in contradiction with the inequality in Eq. (\ref{eq0}). This result proofs the corollary. 
\end{proof}

\begin{corollary}
	There is no differentiable function $w:(0, + \infty) \rightarrow \mathbb{R}$ with continuous derivative at  $(0, \delta)$ for some $\delta>0$ which satisfy:
	\begin{itemize}
		\item[i)]  $\displaystyle{w-\frac{1}{3}\frac{d \ln|w|}{d\ln a} \geq 0};$
		\item[ii)] $-1\leq w(a)$ for all $a>0$;
		\item[iii)]  $ \lim\limits_{a \rightarrow 1}w(a)=w_0$, for some $0>w_0 \geq -1$.
	\end{itemize}
\end{corollary}
\begin{proof}
This corollary is  a straightforward consequence of the previous corollary and of the fact that every differentiable function is continuous.
\end{proof}

This result proofs our assertion that there is no function of $C_1$ class that explain the present time accelerated expansion of the Universe and at same time satisfies the thermodynamical bounds. It should be noted that if the continuity condition is relaxed, it is possible to build functions $w(a)$ that fulfill the conditions i), ii) and iii) above except on a finite set of points where $w(a)$ is not differentiable. However, 
physically speaking, we have no reason to think that $w$ is not differentiable.\footnote{Remember that the case $w=-1$, which mimics the cosmological constant, is in excellent agreement with the data.}. Also, in these discontinuity points the stability condition fails. Thermodynamics do not allows such exceptions.

Nevertheless, it is interesting for the sake of completeness obtain an example of a differentiable function except from a set of finite points that satisfies the conditions i), ii) and iii). As we have seen, in this case we should look for a positive function for $t$ sufficiently small and with at least one discontinuity point in the interval $(0,1)$. In the next proposition we will show a function with these characteristics. 



\begin{proposition}
Let us consider $0>w_0>-1$ and $d=(w_0+1)/(3w_0)$. Let $h:(0,e^{d}) \rightarrow \mathbb{R}$ be a positive, decreasing and differentiable function with $h(a)<1$ for all $a$. By defining the function $w:(0,+\infty) \rightarrow \mathbb{R}$ such as
$$
w(a)= \left\{
\begin{array}{lll}h(a)&if & a\in \left(\left.0,e^{d}\right]\right.\vspace{0.3cm}\\
-\frac{1}{3\ln (a)-\frac{1}{w_0}}&if& a \in (e^{d},+ \infty).
\end{array}
\right.
$$
Then
\begin{enumerate}
	\item $\displaystyle{w-\frac{1}{3}\frac{d \ln|w|}{d\ln a}\geq 0}$ for all $a \neq e^{d};$
	\item $ \lim\limits_{a \rightarrow 1}w(x)=w_0$
	\item $-1\leq w(a)\leq 1$ for all $a> 0.$
\end{enumerate}
\end{proposition}
\begin{proof}
	The function $w$ is decreasing at $(0,e^{d})$. In this case, 
	$$
	\frac{dw}{dt}<0 \leq \frac{w^{2}}{3a}
	$$
	and given that $w>0$ in this interval
$$\displaystyle{w-\frac{1}{3}\frac{d \ln|w|}{d\ln a}> 0}.$$

\noindent for all $a \in (0, e^{d})$. If  $a>e^{d}$, then
$$\displaystyle{w-\frac{1}{3}\frac{d \ln|w|}{d\ln a}= 0}.$$ Therefore the function $w$ satisfies the condition 1. Note that $e^{d}<1$ and that $w$ is continuous at $1$. Thus
$$
\lim\limits_{a \rightarrow 1}w(a)=w(1)=w_0,
$$
that is, the function $w$ satisfy 2. By hypothesis, for $a \in (0, e^{d})$, $0<w(a)<1$. The function $v(a)=3\ln(a)-\frac{1}{w_0}$ is crescent, and $v(e^{d})=1$. Therefore, for all $t>e^{d}$
$$
0>w(t)=\frac{-1}{v(t)}>-1
$$
which proofs 3.
\end{proof}

Note that the discontinuity point of the functions shown in the previous proposition occurs at the point $e^{d}$ such that $e^{d} \in (0,1)$, since $d<0$. 

\begin{corollary}
Let us consider $0>w_0>-1$ and $d=(w_0+1)/(3w_0)$. Let us define the function $w:(0,+\infty) \rightarrow \mathbb{R}$ such as
$$
w(a)= \left\{
\begin{array}{lll}-\frac{a}{2e^{d}}+1&if& a\in \left(\left.0,e^{d}\right]\right.\vspace{0.3cm}\\
\frac{-1}{3\ln (a)-\frac{1}{w_0}}&if& a \in (e^{d},+ \infty).
\end{array}
\right.
$$
Then:
\begin{enumerate}
	\item $\displaystyle{w-\frac{1}{3}\frac{d \ln|w|}{d \ln a}\geq 0}$ for all $a \neq e^{d};$
	\item $ \lim\limits_{a \rightarrow 1}w(a)=w_0$;
	\item $-1\leq w(a)\leq 1$ for all $a> 0.$
\end{enumerate}	
\end{corollary}
\begin{proof}
Note that the function $h(a)=-\frac{a}{2e^{d}}+1$, for all $a\in \left(\left.0,e^{d}\right]\right.$ is positive, decreasing and $h(a)<1$ for all $a$. The result follows from the previous proposition.
\end{proof}

\begin{figure}[t]
\includegraphics[width=3.0truein,height=3.0truein]{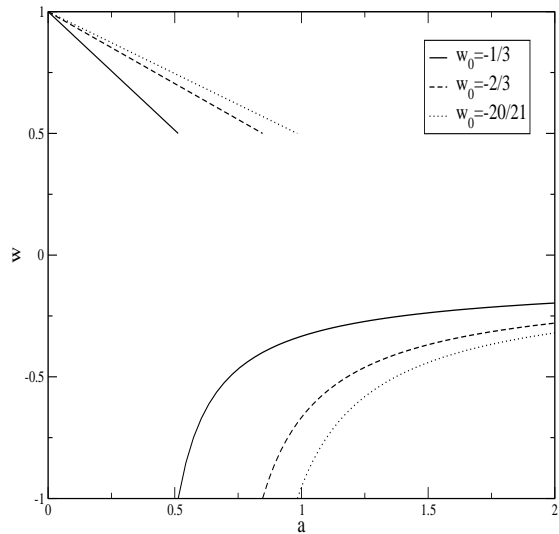}
\caption{Examples of quintessential EoS parameters satisfying the thermodynamical conditions except at a point where $w$ is not differentiable.\label{EoS_example}}
\end{figure}

\noindent This EoS parameter mimics a quintessence scalar field ($-1\leq w(a)\leq1$), but the stability condition fails at just an point: $e^d$. This is the closest definition we can obtain to a DE fluid that mimics a quintessence scalar field since the thermodynamical bounds implies that any DE EoS parameter will have at least one discontinuity point at $(0,1)$. Figure \ref{EoS_example} shows the curves from the above function for $w_0=-1/3,\,-2/3$ and $-20/21$. As we can see, the discontinuity approaches of $a=1$ as $w_0$ approaches of $-1$. In fact, as will be show below, even if we allow that the EoS parameter $w$ to be differentiable except on a finite number of points, it is impossible to build a quintessential like EoS which mimics a cosmological constant at present time ($w(a\to1)\to-1$) and satisfies the thermodynamical constraints.
%


\begin{proposition}
	There is no function $w:(0,+\infty) \rightarrow \mathbb{R}$, such that $1$ is an isolated point from the set of points in which $w$ is differentiable and that satisfies:
	\begin{enumerate}
		\item $\displaystyle{w-\frac{1}{3}\frac{d \ln|w|}{d\ln a} \geq 0};$ at the points where $w \neq 0$ and \\ \\ $w$ is differentiable\,\,;
		\item$ \lim\limits_{a \rightarrow 1}w(a)=-1$\,\,;
		\item $-1\leq w(a)\leq 1$ for all $a>0$.
	\end{enumerate}
\end{proposition}
	\begin{proof}
		Suppose by contradiction that there is a function that fulfill all conditions described in the proposition 3.
		Let $J$ be an interval centered at $1$ such that $w$ is differentiable at
		$J\setminus \left\{1\right\}$.  It is possible to choose this interval since $1$ is a isolated point from the set of points in which $w$ is differentiable. Given that $\lim\limits_{a \rightarrow 1}w(a)=-1$ then we can suppose, without lost of generality, that $w<0$ at $J \setminus \left\{1\right\}$. 
		Therefore, for all $a \in J \setminus \left\{1\right\}$
		$$
		\frac{d \ln|w|}{d\ln a}=a\frac{d \ln(-w)}{da} = \frac{a}{w}\frac{dw}{da}
		$$
		so that
		$$\displaystyle{w-\frac{1}{3}\frac{a}{w}\frac{dw}{da} \geq 0}.$$
		For all $a \in J\setminus\left\{1\right\}$
		$$\displaystyle{w^{2}-\frac{a}{3}\frac{dw}{da} \leq 0}.$$
		that is,
		$$0<\displaystyle{\frac{3w(a)^{2}}{a}\leq \frac{dw}{da}(a)}.$$
		Thus the function $w$ is increasing at $J\setminus \left\{1\right\}$. Let $t_{0}<t_{1}<1$, $t_{0},t_{1}\in J$. We have that
		$$
		-1 \leq w(t_{0})<w(t_{1})< w(a).
		$$
		for all $a \in J$, with $t_{1}<a<1$. Therefore
		$$
		-1=\lim\limits_{a \rightarrow 1^{-}}w(a)\geq w(t_{1})>w(t_{0})\geq -1
		$$
		which is  contradictory. This result proofs the proposition.
	\end{proof}

\begin{corollary}
There is no function $w:(0, + \infty) \rightarrow \mathbb{R}$ differentiable except on a finite set that satisfies:
\begin{enumerate}
	\item $\displaystyle{w-\frac{1}{3}\frac{d \ln|w|}{d\ln a} \geq 0};$
	\item$ \lim\limits_{a \rightarrow 1}w(a)=-1$
	\item $-1\leq w(a)\leq 1$ for all $a>0$.
\end{enumerate}
\end{corollary}

\begin{proof}
In fact, since the set of points where $w$ is not differentiable is finite, the number $1$ would be an isolated point of this set. The corollary follows as a consequence of the previous proposition.
\end{proof}

\section{Conclusions}

By considering the content of the Universe as a perfect fluid, it was shown in Ref. \cite{Barboza_thermo} that the cosmic fluids should necessarily have a thermodynamical stability. In the case of a perfect fluid with an EoS parameter $w=p/\rho$, this implies that the inequality $$w-\frac{1}{3}\frac{d\ln \vert w\vert}{d\ln a}\ge0$$ must be satisfied. In the case of a constant EoS parameter, this inequality implies that $w\ge0$ meaning that a DE fluid with a constant EoS parameter is ruled out by the laws of thermodynamics. Obviously this result does not mean immediately that a time-dependent DE EoS parameter is also ruled out by thermodynamics. Indications that this occurs is provided by Ref. \cite{Barboza_thermo} which shows that the observational constraints upon a DE fluid with a time-dependent EoS parameter are in fact in conflict with the stability conditions. 

In this paper we have continued the ideas presented in Ref. \cite{Barboza_thermo}. The stability condition along with the positiveness of the entropy, which implies that $w\ge-1$, means that an EoS parameter that is able to lead the Universe to the present day accelerated phase, i.e., $w_0=w(a=1)<-1/3$, should comply with the following requirements:
\begin{enumerate}
\item[i)] $\displaystyle{w-\frac{1}{3}\frac{d\ln \vert w\vert}{d\ln a}\ge0}$; 
\item[ii)] $w\ge-1$;
\item[iii)] $w_0<-1/3.$
\end{enumerate}

In this way, we have provided a rigorous demonstration, which showed that there is no differentiable function, $w(a)$, with a continuous derivative that satisfies all three conditions above (Corollary 3). We have shown that these three conditions are true only if $w$ is differentiable except on a finite set, where the  condition i) is not true. Although we have assumed that $w$ is differentiable except on a finite set,  we proved precisely that quintessence scalar fields $-1\leq w\leq1$, which mimics a cosmological constant at present time, i.e., $w(a\to1)\to-1$, were ruled out (Corollary 5). This means that DE fluids are unphysical and should be treated merely as mathematical artifacts to explain the data but, at the same time, carries no physical significance. Although a perfect fluid with $w=-1$ acts within the Einstein field equations in the same way that the sum of the zero point energy of all fields filling the Universe, this mathematical equivalence does not mean any physical equivalence  since the quantum vacuum is not a substance. Therefore, we believe that we have demonstrated precisely that the vacuum energy remains the strongest candidate to explain the current accelerated expansion of the Universe and the cosmological constant problem remains as one of the biggest problems of the theoretical cosmology.

\section*{Acknowledgments}

\noindent E.M.C.A. and J.A.N. thank CNPq (Conselho Nacional de Desenvolvimento Cient\'ifico e Tecnol\'ogico), Brazilian scientific support federal agency, for partial financial support, Grants numbers 302155/2015-5 (E.M.C.A.) and 303140/2017-8 (J.A.N.). 



\begin{thebibliography}{99}

\bibitem{acc_exp1} A. G. Reiss {\it et al.}, Astron. J. {\bf 116}, 1009 (1998).
\bibitem{acc_exp2} S. J. Perlmutter {\it et al.}, Astrophys. J. {\bf 517}, 565 (1999).
\bibitem{acc_exp3} D. N. Spergel {\it et al.}, Astrophys. J. Suppl. {\bf 148}, 175 (2003).
\bibitem{acc_exp4} D. J. Eisenstein {\it et al.}, Astrophys. J. {\bf 633}, 560 (2005).
\bibitem{zeldovich} Ya. B. Zeldovich, JETP Lett. {\bf 6}, 316 (1967); Sov. Phys. Uspekhi {\bf 11}, 381 (1968).
\bibitem{weinberg} S. Weinberg, Rev. Mod. Phys. {\bf 61}, 1 (1989). 
\bibitem{padmanabhan} T. Padamanbhan, Class. Quant. Grav. {\bf 19} L167 (2002); Phys. Rept. {\bf 380}, 235 (2003).
\bibitem{bousso}  R. Bousso, Gen. Rel. Grav. {\bf 40}, 607 (2008).
\bibitem{CCP1} V. Emelyanov and F. R. Klinkhamer, Phys. Rev. D {\bf 85}, 103508 (2012).
\bibitem{CCP2} D. J. Shaw and J. D. Barrow, Phys. Rev. D {\bf 83}, 043518 (2011); J. D. Barrow and D. J. Shaw, Phys. Rev. Lett. {\bf 106}, 101302 (2011). 
\bibitem{CCP3} S. Aslanbeigi {\it et al.}, Phys. Rev. D {\bf 84}, 103522 (2011).
\bibitem{CCP4} P. D. Mannheim, Gen. Rel. Grav. {\bf 43}, 703 (2011).
\bibitem{CCP5} A. Linde and V. Vanchurin,  arXiv:1011.0119.
\bibitem{CCP6} H. Stefancic, Phys. Lett. B {\bf 670}, 246 (2009).
\bibitem{CCP7} J. Garriga and A. Vilenkin, Phys. Rev. D {\bf 64}, 023517 (2001).
\bibitem{CCP8} S. M. Carroll and G. N. Remmen, Phys. Rev. D {\bf 95}, 123504 (2017). 
\bibitem{copeland} E. J. Copeland, M. Sami and S. Tsujikawant, Int. J. Mod. Phys. D {\bf15}, 1753 (2006).
\bibitem{quintessence} C. Wetterich, Nucl. Phys. B {\bf 302}, 668 (1988); B. Ratra and  P. J. E. Peebles, Phys. Rev. D {\bf 37}, 3406 (1988); R. R. Caldwell, R. Dave and P. J. Steinhardt, Phys. Rev. Lett. {\bf 80}, 1582 (1998); R. R. Caldwell, Braz.\ J.\ Phys.\ {\bf 30}, 215 (2000).
\bibitem{phantom} R. R. Caldwell, Phys. Lett. B {\bf{545}}, 23 (2002); R. R. Caldwell, M. Kamionkowski and N. N. Weinberg, Phys. Rev. Lett. {\bf{91}}, 071301 (2003); S. M. Carroll, M. Hoffman and M. Trodden, Phys. Rev. {\bf{D68}}, 023509 (2003).
\bibitem{Barboza_thermo} E. M. Barboza Jr., R. C. Nunes, E. M. C. Abreu and J. A. Neto, Phys. Rev. D {\bf 92}, 083526 (2015).
\bibitem{weinberg1971} S. Weinberg, Astrophys. J. {\bf 168}, 175 (1971).
\bibitem{entropy_positiveness} J. A. S. Lima and J. S. Alcaniz, Phys. Lett. B {\bf 600}, 191 (2004); H. H. B. Silva, R. Silva, R. S. Gon\c{c}alves, Zong-Hong Zhu and J. S. Alcaniz, Phys. Rev. D {\bf88}, 127302 (2013).
\bibitem{callen} H. B. Callen, {\it Thermodynamics and an Introduction to Thermostatiscs} (2nd edition, Wiley, New York, 1985).

\bibitem{elonlages}    E. L. Lima, {\it An\'alise Real}, (10th edition, IMPA, RJ, 2009).

 
\end{thebibliography}
\end{document}